\documentclass{llncs}
\usepackage{amsmath,amssymb,amsfonts}
\usepackage{graphicx}
\usepackage{algorithmic}
\usepackage{algorithm}
\usepackage{sober}

\begin{document}

\mainmatter  

\title{Fast error-tolerant quartet phylogeny algorithms}
\titlerunning{Fast quartet phylogeny algorithms}
\author{Daniel G. Brown \and Jakub Truszkowski}
\authorrunning{Daniel G. Brown and Jakub Truszkowski}
\institute{David R. Cheriton School of Computer Science\\
University of Waterloo\\
Waterloo ON N2L 3G1 Canada\\
\email{{browndg,jmtruszk}@uwaterloo.ca}}

\maketitle

\begin{abstract}
We present an algorithm for phylogenetic reconstruction using quartets that returns the correct topology for $n$ taxa in $O(n \log n)$ time with high probability, in a probabilistic model where a quartet is not consistent with the true topology of the tree with constant probability, independent of other quartets.  Our incremental algorithm relies upon a search tree structure for the phylogeny that is balanced, with high probability, no matter what the true topology is.  Our experimental results show that our method is comparable in runtime to the fastest heuristics, while still offering consistency guarantees.
\end{abstract}
\newpage
\section{Introduction}

Incremental phylogenetic reconstruction algorithms add new taxa to a topology until all $n$ taxa have been added.  They optimize a greedy objective at all $n$ insertions, much as agglomerative algorithms (like neighbour joining or UPGMA) optimize an objective at all $n-1$ agglomerations.  Such algorithms can be quite efficient.  If each addition requires $O(f(n))$ time, the overall runtime is $O(nf(n))$.  

We give an algorithm where each insertion requires $O(\log n)$ runtime with high probability, and where the probability that any insertion is incorrect is $o(1)$ in a simple error model.  Thus, our randomized algorithm has runtime $O(n\log n)$ with high probability (regardless of the true topology) and $o(1)$ probability of producing an incorrect topology.   We believe it is the first $O(n poly \log n)$-runtime algorithm with such guarantees.  Any $o(n\log n)$-runtime algorithm cannot return all topologies, so our algorithm is asymptotically optimal.


We present a review of related work,  give basic definitions, and then give the algorithm in the case of error-free data.  Then, we extend the algorithm to the case of data containing noise.  Finally, we give some experimental results on real and simulated data.  Our error-tolerant algorithms offer the possibility of producing a phylogenetic tree in runtime smaller than that of producing even the input matrix to a distance method like neighbour joining, while still having high probability of reconstructing the true tree.

\section{Related work}

Phylogenetic quartet methods reconstruct trees from sets of four taxa and combine these phylogenies into the overall tree. Quartet puzzling~\cite{QP} is one of the first algorithms in this line of research. Many heuristic algorithms also operate on this principle ({\em e.g.}~\cite{gascuel,short-qp}).

Some quartet algorithms find the correct phylogeny with high probability under a certain model of evolution. Erd{\"o}s {\em et al.}~\cite{erdos} give an $O(n^4\log n)$ algorithm that reconstructs the phylogeny with $1-o(1)$ probability assuming that the sequences evolve according to the Cavender-Farris model of evolution, for sufficiently long sequences. The runtime of their algorithm is $O(n^2)$ for most trees. Cs{\H u}ros~\cite{csuros} provided a practical $O(n^2)$ algorithm with similar performance guarantees. Recent papers~\cite{gronau,mossel} give similar algorithms to identify parts of the tree that can be reconstructed. These approaches choose queries so that, in the assumed model of evolution, all queries are  correct with high probability. 

The only sub-quadratic time algorithm with guarantees on reconstruction accuracy is by King {\it et al.}~\cite{king}. The running time is $O(n^2\frac{\log \log n}{\log n})$ provided that the sequences are long enough.

Wu {\it et al.}~\cite{Kao} gave a simple error model where each quartet query independently errs with fixed probability $p$. They gave an $O(n^4 \log n)$ algorithm that errs with constant probability under this model. This model has also been used for evaluating algorithms for maximum quartet consistency~\cite{Lin}.

We improve on Wu {\it et al.} in runtime and accuracy with an $O(n\log n)$ algorithm that errs with probability $o(1)$. To our knowledge, it is the first provably error-tolerant, substantially sub-quadratic time algorithm for phylogenetic reconstruction.  (Recently, an $O(n^{1.5})$ heuristic algorithm has been proposed~\cite{FastTree}.)

Fast algorithms have been proposed for error-free data. Kannan {\it et al.}~\cite{Kannan} use error-free {\it rooted triples} in an $O(n\log n)$ algorithm. Rooted triples reduce to quartets if we pick one taxon as an outgroup and always ask quartet queries for sets with that taxon, so that algorithm works for error-free quartets.

Our algorithm uses ideas from work on noisy binary search in which comparisons have fixed error probability, by Feige {\it et al.}~\cite{Feige} and Karp and Kleinberg~\cite{Karp}.

\section{Definitions}

We begin with definitions about the two trees we will focus on: the phylogeny we are reconstructing and the search tree that allows us to do the insertions. 

A {\em phylogeny} $T$ is an unrooted binary tree with $n$ leaves in 1-to-1 correspondence with a set $\mathcal{S}$ of taxa.  Removing internal node $v$, and its incident edges, from a phylogeny yields three subtrees, $t_i(T,v)$ for $i=1,2,3$.  The tree $t_i(T,v)$ joined with its edge to $v$ is the {\em child subtree} $c_i(T,v)$.  Phylogeny $T'$ is {\em consistent} with $T$ if its taxa are a subset of those of $T$, and $T'$ is formed by the union of all paths in $T$ between taxa in $T'$, with internal nodes of degree 2 removed.  A {\it border node} of subtree $T'$ in $T$ is any internal node of $T$ that is a leaf in $T'$.   

A {\em quartet} is a phylogeny of four taxa.  A {\em quartet query} $q(a,b,c,d)$, returns one of three possible quartet topologies: $ab|cd$, $ac|bd$ and $ad|bc$, where in $ab|cd$, if we remove the internal edge, we disconnect $\{a,b\}$ from $\{c,d\}$.  We assume a quartet query can be done in $O(1)$ time.  In Section \ref{sec-errors} our error model considers how often quartet queries for four taxa of $T$ are inconsistent with $T$.
A {\em node query} $N(T,v,x)$ for internal node $v$ of phylogeny $T$ and new taxon $x$ is a quartet query $q(x,a_1,a_2,a_3)$, where $a_i$ is a leaf of $T$ in $t_i(T,v)$.  Such a query identifies the  $c_i(T,v)$ where taxon $x$ belongs, if it is consistent with the true topology.  
\subsection{Search tree}

A natural algorithm to add taxon $x$ to phylogeny $T$ begins at an internal node $v$ and uses node query $N(T,v,x)$ to identify the $t_i(T,v)$ where taxon $x$ belongs.  We move to the neighbour of $v$ in that subtree, and repeat the process until the subtree into which $x$ is to be placed is only one edge $e$, which we break into two edges and hang $x$ onto; see Figure \ref{fig-simplest}.  We follow the path from $v$ to an endpoint of $e$ and identify the other endpoint with one more query.  The number of node queries equals this path length plus one.  For a balanced tree with diameter $\Theta(\log n)$, this gives a $\Theta(n\log n)$ incremental phylogeny algorithm.  But for trees like a caterpillar tree, with $\Theta(n)$ diameter, this algorithm requires $\Theta(n^2)$ queries.

\begin{figure}[htb]
\centering
\includegraphics[width=3in]{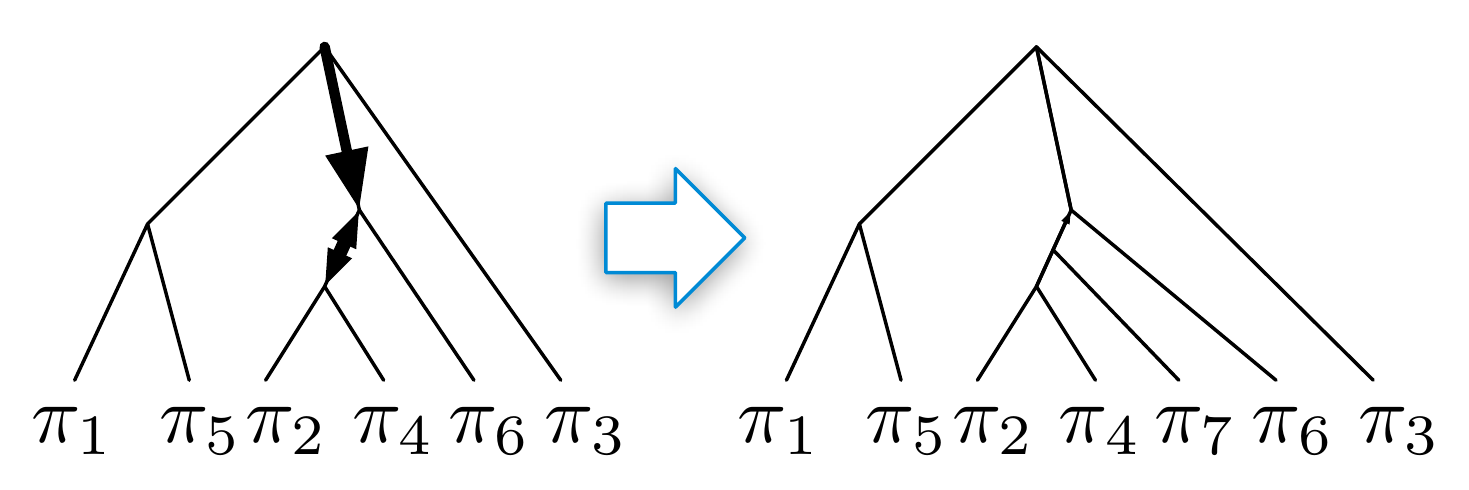}
\caption{Natural incremental algorithm: start at root and search to find place for new taxon $\pi_7$ by asking queries down the path.  Break an edge to insert the new taxon.
\label{fig-simplest}}
\end{figure}

We give a search tree structure to manage the expected number of queries on the search path, regardless of the underlying tree topology.

\begin{figure}[t]
\centering
\includegraphics[width=2.3in]{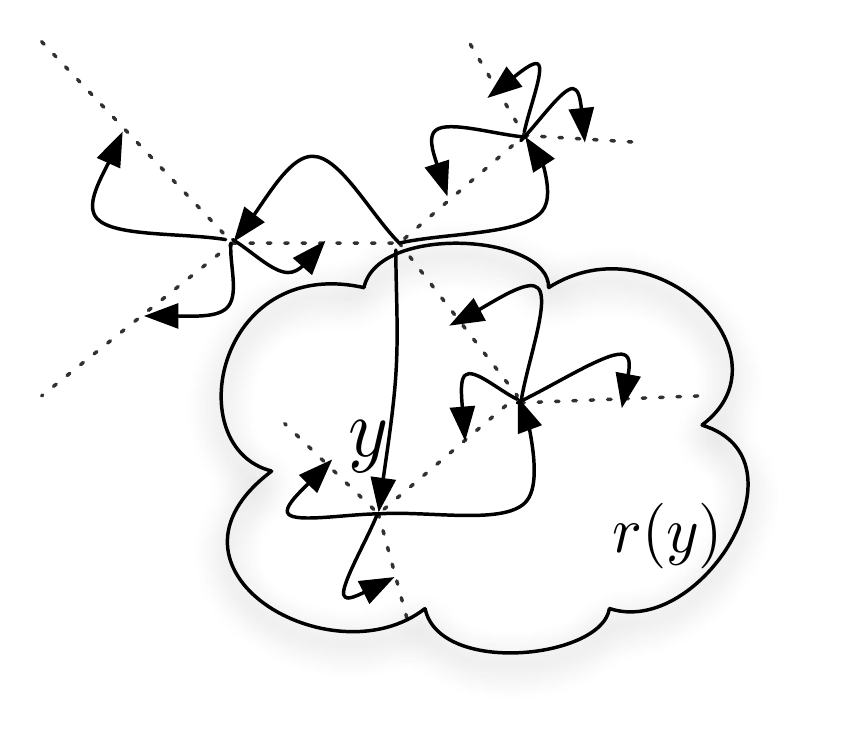}
\caption{A search tree for a seven-taxon phylogeny.  Directed search tree edges are shown in solid lines; the underlying phylogeny is in dotted lines.  The search tree node $y$ corresponds to the region $r(y)$ of the phylogeny indicated by the cloud.
\label{fig-searchtree}}
\end{figure}

\begin{definition}
A {\it search tree} $Y(T)$ for a phylogeny $T$ is a rooted ternary tree satisfying the following conditions:
\begin{enumerate}
\item Each node $y$ in $Y(T)$ is associated with a distinct subtree $r(y)$ of $T$.
\item The root of $Y(T)$ is associated with the full tree $T$.
\item For each internal node $y$ in $Y(T)$, there exists an internal node $s(y)$ in $T$ such that the three subtrees associated with the children of $y$ are the intersections between $r(y)$ and the three child subtrees of the node $s(y)$ in $T$.  There are also three nonempty lists $\ell_i(y)$ stored at each internal node $y$; each element of $\ell_i(y)$ is a taxon in $t_i(T,y)$.
\item For each node $y$ in $Y(T)$, $r(y)$ has at most two border nodes in $T$
\end{enumerate}
\end{definition}

$Y(T)$ is {\em complete} if each leaf in $Y(T)$ is associated with a single edge of $T$, and each edge of $T$ has a corresponding leaf in $Y(T)$.  For a given node $y$ in the search tree, its associated node $s(y)$ in $T$ may be picked so the three child subtrees are reasonably balanced; this gives expected $O(\log n)$ insertion time.  See Figure \ref{fig-searchtree} for an example.

\section{An algorithm for error-free data}

Using our search tree structure gives a straightforward incremental phylogeny algorithm if quartets are all consistent with $T$, the true topology.  

We pick a random permutation $\pi$ of the taxa, and start with the unique topology $T_3$ for $\{\pi_1,\pi_2,\pi_3\}$, and a search tree $Y(T_3)$ with four nodes: a root $w$ with $r(w) = T_3$ and $s(w)$ the internal node of $T_3$, and with one leaf for each edge of $T_3$.  We also store $\ell_i(w) = \{\pi_i\}$; we also use $\ell_i(w)$ to represent the unique member of this set. This fits our requirements for a complete search tree of $T_3$.

Now, assuming $T_i$ is consistent with $T$, and $Y(T_i)$ is a valid search tree for $T_i$, we add $\pi_{i+1}$, to produce $T_{i+1}$ and $Y(T_{i+1})$.  We start at the root $w$ of $Y(T_i)$ and ask the node query $N(T_i,s(w),\pi_{i+1})$ using the quartet $q(\pi_{i+1},\ell_1(w),\ell_2(w),\ell_3(w))$; this tells us which child of $w$ we should move to next.  We continue until we reach a leaf $y$ of $Y(T_i)$; this corresponds to the edge $e$ of $T_i$ where the new taxon $\pi_{i+1}$ belongs.  We break edge $e$ into two parts, creating a new node $u$ and a new edge from $u$ to the new leaf $\pi_{i+1}$.  The new tree is $T_{i+1}$.  

To update $Y(T_i)$, we create three edges from $y$ to a new node for each of the three newly created edges and let $\ell_1(y)$ be $\{\pi_{i+1}\}$, and set $\ell_2(y)$ and $\ell_3(y)$ to contain the taxon closest to $\pi_{i+1}$ in the final quartet query and one of the two taxa that was not closest to $\pi_{i+1}$ in that query.   Since node $y$ was a leaf in $Y(T_i)$, these nodes are in proper configuration with respect to $y$ in $T_{i+1}$.  See Figure \ref{fig-searchtree-insert}.

\begin{figure}[t]
\centering
\includegraphics[width=4in]{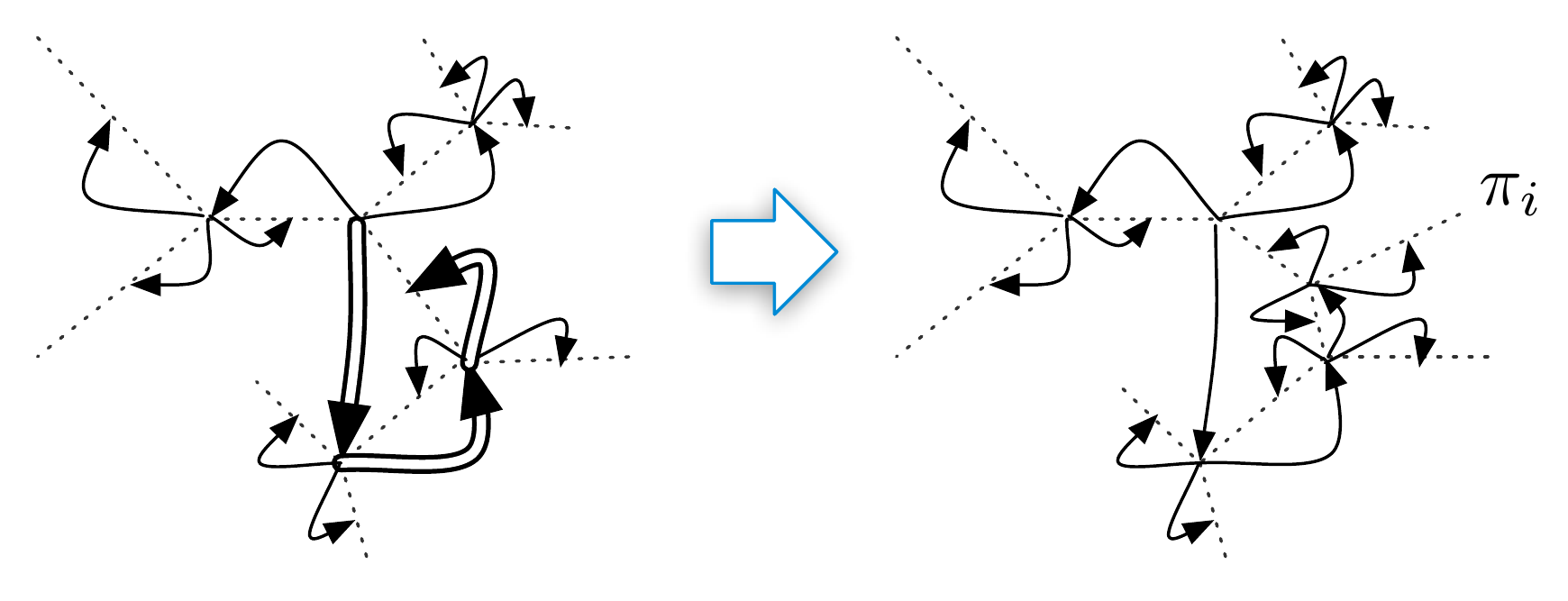}
\caption{Inserting into a search tree.  To insert $\pi_8$ into the phylogeny, we follow the path through the search tree indicated with double arrows.  We find the correct edge to break to add $\pi_8$ to the tree, and modify the search tree locally to accommodate the change.
\label{fig-searchtree-insert}}
\end{figure}

Assuming the quartet queries all are consistent with the true topology $T$, we discover in this way the proper place in the tree to insert each new taxon and maintain the invariants required for a complete search tree.  In particular, the only subtrees whose border nodes need to be considered are those created by the new node addition, and as they are all either single edges or derived from a single edge in $Y(T_i)$, they continue to have at most two border nodes.

\begin{theorem} 
\label{thm-no-errors}
If all quartet queries made by this algorithm are consistent with $T$, then this algorithm returns $T$.  Its runtime is $O(n \log n)$ with probability $1-o(1)$.  
\end{theorem}

\begin{proof}
We have seen that the algorithm returns $T$. In the next subsection, we show that inserting  taxon $\pi_i$ requires $O(\log n)$ queries with high probability, each of which requires $O(1)$ time; the work to create a new edge requires constant time.  The overall runtime is $O(n\log n)$ with high probability.
\end{proof}

\subsection{The height of the search tree}
To prove Theorem \ref{thm-no-errors}, we need to know the height of the search tree $Y(T)$.  We will show that this tree is almost surely balanced, using several lemmas.

\begin{lemma}
For any phylogeny $T$, with $n$ taxa, there exist two disjoint child subtrees $A$ and $B$ of the form $t_i(T,v)$ with at least $n/6$ and at most $n/3$ taxa.
\end{lemma}

\begin{proof}
We first show there exists a node $u$ where all $t_i(T,u)$ have at most $n/2$ taxa.  Pick an internal node $u$ in $T$; if all $t_i(T,u)$ have at most $n/2$ taxa, we are done.  Otherwise, move to the its neighbour in the $t_i(T,u)$ with the most taxa.  This process terminates at a node $u$ satisfying the property. Let $n_1 \le n_2 \le n_3$ be the numbers of taxa in the trees $t_i(T,u)$ at some step. If $n_1>\frac{n}{2}$, we move to the neighbour $u^*$ of $u$ in $T_1$; trees $T_i(T,u^*)$ have $n_{11},n_{12}$ and $n_2+n_3$ taxa respectively where $n_{11},n_{12}$ are the numbers of taxa in the subtrees $T_{11},T_{12}$ of $T_1$ created by removing $u^*$. Since $n_2+n_3<\frac{n}{2}$, the component with size over $\frac{n}{2}$ must be either $T_{11}$ or $T_{12}$, which are smaller than $T_1$  since they are its subtrees.

Now, consider the node $u$ we have found by this process, and let $t_1$ and $t_2$ be the two largest $t_i(T,u)$ subtrees, both of which have between $n/4$ and $n/2$ taxa.  If $t_1$ has more than $n/3$ taxa, consider the three child subtrees in $t_1$ of the neighbour of $u$ in $t_1$; one has zero taxa, so the larger must have at least $n/6$ taxa.  If this tree has at most $n/3$ taxa, we have found our subtree $A$; if not, we move one step more away from $u$ until we find a subtree small enough.  We analogously find $B$ as a subtree of $t_2$.
\end{proof}

\begin{lemma}
\label{tree-height-blah}
The number of node queries asked by the phylogeny algorithm to assign taxon $\pi_{i+1}$ to its place in the tree is at most $37(\log_{6/5}i) \approx 203 \ln i$, with probability $1-o(1/i^4)$, and at most $37(\log_{6/5}n)$ with probability $1-o(1/n^4)$.
\end{lemma}

\begin{proof}
Consider the process of adding $\pi_{i+1}$ to the tree.  We consider a sequence $y_1\ldots y_k$ of nodes in the search tree $Y$, each corresponding to a subtree $r(y_j)$ of the existing phylogeny.  We divide the $y_j$ into phases: phase $t$ corresponds to the period in which $r(y_j)$ contains between $\frac{5}{6}^ti$ and $\frac{5}{6}^{t-1}i$ leaves; after $\log_{6/5}i$ phases, the algorithm has found where to put $\pi_i$.  We show that the distribution of the length of each phase is bounded above by the sum of three geometrically-distributed random variables.

Each phase corresponds to taking a subtree and shrinking it by a factor of 5/6.  This happens either if the largest of the three subtrees of the phylogeny descendant from the current search tree node $y_j$ has at most 5/6 of the number of taxa we had at the beginning of the current phase, or if $\pi_i$ belongs in a tree with fewer than that many taxa.  We concern ourselves only with the first of these ways of ending a phase, so we upper bound the length of a phase.  

The queries asked include taxa found in $r(y_j)$, in the order that they occur in permutation $\pi$.  In particular, we will ask a node query including a node of $A$ with probability at least $1/6$ at step, independently, until we finally do ask a query of a node from $A$.  (Since our queries always include at least $5/6$ of the taxa, and we have not queried any members of $A$, we always have all members of $A$ available.)  After querying a member of $A$, for the phase to continue, we must choose the subtree containing all of $B$.  Now, we ask queries corresponding to the current subtree, until we see a taxon from $B$, which will happen with probability $1/6$ or greater at each step.  Now, we arrive in a state where the current subtree of the phylogeny includes border nodes inside $A$ and $B$, since we must have cut off parts of $A$ and of $B$, but cannot have cut off all of either without ending the phase.  Now, we ask queries until we see a node from neither $A$ nor $B$; this happens with probability at least $1/3$ at each step. Then, the current search tree node $y_j$ must correspond to a node on the edge from $A$ to $B$ in the phylogeny, since otherwise one of its subtrees would have three border nodes.  

Thus, the length of a phase is at most the sum of three geometric random variables, with expectations $6$, $6$ and $3$; we then move to a new tree with at most $5/6n$ taxa.  However, it may have two border nodes as well; we label these with a taxon from their neighbouring subtrees (thereby adding two taxa to the current subtree) and perform a single quartet query (removing at least two taxa).   This gives a new subtree in which we can perform the next phase.

Thus, if $G(i)$ are independent geometric random variables with mean $i$, then the length of one phase is bounded above by $G(6)+G(6)+G(3)+1$, and the expected total number of queries is at most $19(\log_{6/5}i)+1$, where for simplicity, we let the $G(i)$ all have mean 6.  

Moreover, this variable is rarely above $37 \log_{6/5}i$.  In particular, let $Q(n,r)$ be the negative binomial random variable that is the sum of $n$ geometrically distributed variables with mean $r$.  Then $\Pr[Q(n,r) > knr] = \Pr[B(knr,1/r) < n]$, where $B(n,p)$ is a binomial random variable that results from the sum of $n$ independent Bernoulli trials, each with mean $p$.  By standard Chernoff methods (\cite{Panconesi}, p. 6), this probability is bounded above by $\exp(\frac{-kn(1-1/k)^2}{2r})$.  So, $\Pr[Q(3\log_{6/5}i,6)>36\log_{6/5}i] \le i^{-4}$, meaning that the probability we use more than $37 \log_{6/5} i$ queries for taxon $\pi_{i+1}$ is $o(1/i^4)$; similarly, the probability that we use more than $37 \log_{6/5}n$ queries for taxon $\pi_{i+1}$ is $o(1/n^4)$.
\end{proof}

We emphasize that $Y(T)$ is almost surely balanced regardless of the topology of $T$. Even if the diameter of $T$ is $\Theta(n)$, its corresponding search tree almost surely has height $O(\log n)$. We conjecture that the actual values of the constants are much smaller than mentioned in the above lemma.

\section{Accounting for errors}
\label{sec-errors}

Our search tree algorithm adapts to the case of error-prone quartets where each quartet query independently errs with probability $p>0$. We assume that $(1-p)^3>0.5+\epsilon$ for some $\epsilon>0$; we relax this assumption at the end of the section.

\subsection{Random walk in the search tree}

Let $Y(T')$ be a complete search tree for $T'$ and let $x$ be a taxon not in $T'$.  We will perform a random walk on $Y(T')$ to place $x$ into its proper place in $T'$, where each step of the random walk is determined by at most 3 quartet queries.

Let $y_i$ be the location of the random walk after $i$ steps, with $y_0$ the root of $Y(T')$. If $y(i)$ is not a leaf node, query the border nodes of $r(y_i)$. If any border node queries gives answer $x \notin r(y_i)$, go to the parent node of $y_i$. If all border nodes give answers consistent with $x \in r(y_i)$, query the node $y_i$ and descend to the child of $y_i$ indicated. 

If $y_i$ is a leaf, corresponding to an edge of $T'$, let it have counter variable $c$ initially set to $0$. Query its border nodes as before; if each is consistent with $x \in r(y)$, increment $c$. Otherwise, decrement it if it is greater than $0$; if $c =0$, move to the parent node of $y$.  After a number of queries we will soon compute, we are at a node in $Y(T')$: if it is a leaf, add $x$ to that node of the search tree as for the insertion algorithm with error-free data.  If not, signal failure.

The algorithm finds the proper place in the tree with high probability.  Let $y_x$ be the leaf in the search tree where we should insert taxon $x$.  After $i$ steps in the random walk, let the random variable $d_i$ be the distance in the search tree between $y_i$ and $y_x$.  Let the random variable $g_i$ have value $-c$ if $y_x = y_i$, $d_i+c$ if $y_x \neq y_i$ and $y_i$ is a leaf of $Y(T')$, and $d_i$ if $y_i$ is not a leaf.  If $g_i \le 0$, then the current node of the random walk is the correct place to put $x$. The following simple observation is essential to proving the correctness of our algorithm.

\begin{lemma}
\label{chernoff}
Consider the random variables $g_i$ defined above.  
\begin{enumerate}
\item $E[g_i] \leq d_0+(1-2(1-p)^3)i$.
\item If $i>\frac{-d_0}{1-2(1-p)^3}$, then $\Pr [y_i \neq y_x] < \exp(\frac{-(d_0+i(1-2(1-p)^3)^2}{2i})$
\end{enumerate} 
\end{lemma}

\begin{proof}
At each step of the random walk, there are at most two border nodes, so at most three queries. If each gives a correct answer, $g_i$ decreases by 1; if any incorrect queries occur $g_i$ increases by at most one, though it might still decrease by 1. In the worst case, the probability that $g_i$ decreases is at least $(1-p)^3$, so 
$E[g(i+1)-g(i)]\le-(1-p)^3+(1-(1-p)^3)=1-2(1-p)^3$.
The result follows from linearity of expectation, since $g_0 = d_0$.
The second claim follows from the Chernoff bound, as the queries are independent.
\end{proof}

Now, we have a straightforward taxon insertion algorithm.  For each taxon $\pi_{i+1}$, we run the random walk long enough to handle the case that $g_0 = 203 \ln i$.  To make the error probability at most $(1/i^2)$, we require that the random walk have $j$ steps, where  $\exp(\frac{-(203\ln i +j(1-2(1-p)^3))^2}{2j}) \le \frac{1}{i^2}$.  The minimum value of $j$ to make this guarantee is  $j \ge k \ln n$, for $k = \frac{-203(1-2(1-p)^3)+2+2\sqrt{1-203((1-2(1-p)^3)}}{(1-2(1-p)^3)^2}$.

We can now state the taxon insertion procedure in detail.

\begin{algorithm}

\caption{InsertTaxon($x,T,Y(T)$)}
\label{a1}
\begin{algorithmic}
\STATE Initialize the random walk at the root of $Y(T)$.
\FOR{$i=1$ to $k \log n$}
	\STATE Simulate the next step of the random walk.
\ENDFOR
\STATE Let $y_{k \log n}$ be the current node of the random walk.
\IF{$y_{k \log n}$ is a leaf}
	\STATE Attach $x$ to $r(y_{k \log n})$ in $T$ and update $Y(T)$.
\ELSE
	\RETURN Failure.
\ENDIF

\end{algorithmic}
\end{algorithm}

Assuming that the tree $T_{i-1}$ is correct, then, this algorithm adds a new taxon in $O(\log i)$ queries, with error or failure probability $O(1/i^2)$.  

\subsection{Finding quartets to ask}
We must ensure that we can always find a quartet that has not been queried before in $O(1)$ time.  This requires two separate conditions to hold: first, that enough such quartets exist, and second, that we can find them in $O(1)$ time.

The first of these is easy, as long as we start with a constant-sized guide tree $T_S$ on a set $S$ of at least $m$ taxa, where $m$ is the smallest number such that $k \log m < m-2$, with $k$ equal to the multiple of $\log i$ found using the formula in the previous section.  In each insertion phase, we use at most $k \log i$ quartets at any node of the search tree; the extreme case is where the three child subtrees of the current tree $T$ have 1, 1, and $i-2$ taxa in them.

The latter is more complicated. Assume that for each node $y$ in $Y$, $\ell_j(y)$ is the list of all taxa in the child subtree $t_j(r(y),s(y))$ (for $j=1,2,3$). To find the next quartet in $O(1)$ time, we must fetch the next taxon in $t_j(T,s(y))$ in $O(1)$ time. We first enumerate taxa in $\ell_j(y)$. Once all taxa in $\ell_j(y)$ have been used, we pick the border node $b_j(y)$ of $y$ in $t_j(T,s(y))$ (if it exists). The node $b_j(y)$ is associated with some ancestor $y_1$ of $y$ and we have $r(y) \subseteq t_i(r(y_1),b_j(y))$ for some $i$. Taxa in $\ell_{(i+1) mod 3}(y_1) \cup \ell_{(i+2) mod 3}(y_1)$ are also in $t_j(T,s(y))$ so we enumerate them. Once they have been used, we find border nodes of $r(y_1)$ such that two of their taxa lists contain taxa in $t_j(T,s(y))$ that have not been used so far. Once all taxa from a node $y_i$ have been used, we look at border nodes of $r(y_i)$. This process can be thought of as breadth first search on a directed graph where an arc denotes the relationship of being a border node. We leave details to the longer version of this paper.



Now, we give the complete algorithm.  First, pick a constant-sized set $S \subset {\mathcal S}$ of $m$ taxa and find the phylogeny for $S$ consistent with the most quartets. Then iteratively add taxa to the tree using the procedure InsertTaxon described above.

\begin{algorithm}

\caption{Reconstruct($\mathcal{S}$,m)}
\label{a2}
\begin{algorithmic}
\STATE Pick a subset $S \subset \mathcal{S}$ with $m$ taxa
\STATE Find phylogeny $T$ on $S$ consistent with the most quartets by exhaustive search.
\STATE Build a search tree $Y(T)$ for $T$.
\FORALL{$s \in \mathcal{S} \backslash S$}
	\STATE insertTaxon(s,T,Y(T))
\ENDFOR

\end{algorithmic}
\end{algorithm}

The running time of this algorithm is $O(n \log n)$ with high probability. The error probability can be bounded by $\mu(m)+\sum_{i=m}^n \frac{1}{i^2}$, where $\mu(m)$ is the probability that the maximum quartet compatibility tree on a random set of $m$ taxa is not consistent with $T$.  This quantity is constant for constant $m$; in the next section we show how to make the total error probability $o(1)$ as $n$ grows.

The remaining case where $(1-p)^3 \leq \frac{1}{2}$ can be solved by redefining node queries. Each node query is now implemented by asking $c_p$ queries and returning the majority direction, with constants $c_p$ and $C$ chosen appropriately. We defer details to the longer version of this paper. 
\section{Shrinking the error probability to $o(1)$}

The algorithm presented in the previous section errs with constant probability, since it starts with a constant-sized tree that may have errors, and since the additions to this tree also have constant probability of error.

If we start with a non-constant-sized guide tree, we can reduce the error probability.  The main lemma is in the next subsection.

\begin{theorem}
The algorithm Reconstruct$(\mathcal{S},\max(\lceil \log\log n \rceil,m)$ both returns the correct tree and runs in $O(n \log n)$ time with probability $1-o(1)$.
\end{theorem}

\begin{proof}
The exhaustive search step requires enumerating all $O((\log\log n)^4)$ quartets, on all $O((\log \log n)!\log n)$ topologies on $\log \log n$ taxa; the product of these is $O((\log \log n)^{4 + \log \log n}\log n)$, which is sublinear in $n$.  We have already shown that the rest of the algorithm requires $O(n \log n)$ time with high probability.

We will show below that $\mu(\log \log n)$, the failure probability of the guide tree algorithm, is $o(1)$.  The failure probability of the insertion procedure is at most $\sum_{i=\log\log n}^{n} \frac{1}{i^2}$, which is $O(\frac{1}{\log\log n})$, and so $o(1)$.  As such, the overall failure probability is $o(1)$, as desired.
\end{proof}

We note that the guide tree could have more or fewer than $\log \log n$ taxa; we merely require that the brute force guide tree construction requires $O(n \log n)$ time and has $o(1)$ error probability.

\subsection{Maximum quartet consistency is consistent}

Here, we show that the {\it maximum quartet consistency} approach is consistent for our error model.  This result (which may be of independent interest, as our error model has been studied before \cite{Lin}), shows that $\mu(n) \rightarrow 0$ as $n$ grows.

\begin{theorem}
\label{consistency-thm}
Let $T_{mqc}$ be the phylogeny compatible with the most quartet queries for a set of $n$ taxa and let $T^*$ be the true phylogeny. If each quartet query errs independently with probability $p$, then $\mu(n) = \Pr[T_{mqc}\neq T^*]  = o(1)$ as $n \rightarrow \infty$.
\end{theorem}

To prove this theorem, we first show a few properties of quartets.

\begin{definition}
The quartet distance $d_Q(T,T')$ of phylogenies $T$ and $T'$ on the same set of taxa is the number of quartets on which $T$ and $T'$ differ.
\end{definition}

This distance was studied in~\cite{brodal,bryant} among others.

\begin{lemma}
\label{dist-lemma}
The quartet distance between distinct phylogenies is at least $n-3$. 
\end{lemma}
\begin{proof}
Let $T$ and $T'$ be distinct phylogenies. Let $(S_1,S_2)$ be a split in $T$ not present in $T'$. Let $(S_1',S_2')$ be a split in $T'$ not present in $T$ where none of the sets $A=S_1 \cap S_1',B=S_1 \cap S_2', C=S_2 \cap S_1', D=S_2 \cap S_2'$ is empty; such a split exists since $T$ and $T'$ are distinct. Choose taxa $a,b,c,d$ from sets $A,B,C,D$, respectively.  The quartet induced by $T$ is $ab|cd$, whereas in $T'$ it is $ac|bd$.  This gives $\phi=|A||B||C||D|$ conflicting quartets; $\phi$ is at least $n-3$ since $|A|+|B|+|C|+|D|=n$, and the product is minimized when $|A|=n-3$ and $|B|=|C|=|D|=1$.
\end{proof}

The number of trees with  small quartet distance from a fixed tree $T$ is small. 

\begin{definition}
A {\it taxon reinsertion} (TR) operation consists of deleting a taxon from a phylogeny and attaching it to a remaining edge, creating three new edges.
\end{definition}

\begin{lemma}
Let $T$ and $T'$ be phylogenies such that $d_Q(T,T')<n \log^2 n$. The number of TR operations required to transform $T$ into $T'$ is at most $c\log^4 n$ for some constant $c$.
\end{lemma}

\begin{proof}
Let $(S_1,S_2)$ be a split of $T$ not present in $T'$. Let $(S_1',S_2')$ be some split in $T'$ that is not present in $T$ that minimizes $\phi=|A||B||C||D|$ as defined earlier. Without loss of generality, assume that $A$ is the largest of the sets. Observe that each of the sets $B,C,D$ must have at most $\log^2 n$ taxa: otherwise $\phi>n \log^2 n$, so $d_Q(T,T')>n \log^2 n$. We delete all taxa in $B$ and $C$ from both $T$ and $T'$ to create trees $T^{(1)}$ and $T'^{(1)}$. By Lemma \ref{dist-lemma}, this erases at least $n-3$ conflicting quartets. We pick splits $(S_1,S_2)$ and $(S_1',S_2')$ in $T^{(1)}$ and $T'^{(1)}$ as we previously did for the original trees and repeat the process to obtain trees $T^{(2)}$ and $T'^{(2)}$, this time removing at least $n-2\log^2 n-3$ discordant quartets. 

We iterate the process until $T^{(i)}=T'^{(i)}$ for some $i$, which is  $O(\log^2 n)$ since the total number of conflicting quartets is at most $n \log^2 n$, and each iteration erases $\Omega(n)$. The sets $B$ and $C$ have at most $\log^2 n$ taxa at each step of the algorithm. Therefore, at most $O(\log^4 n)$ taxa are deleted from both trees.  

Let $R$ be the taxa removed. The restrictions of both $T$ and $T'$ to ${\mathcal S}-R$ are the same.  To transform $T$ to $T'$, we move all nodes in $R$ to a new side of the tree $T$, and then move each to the proper place in $T'$ in $O(\log^4 n)$ TR operations.
\end{proof}

\begin{corollary}
For any phylogeny $T$, the number of phylogenies $T'$ such that $d_Q(T,T')<n \log^2 n$ is at most $n^{b\log^4 n}$ for a large enough constant $b$.
\end{corollary}

\begin{proof}
Each $T'$ with distance from $T$ at most $n \log^2 n$ can be obtained from $T$ by $c \log^4 n$ TR operations. For any tree, the number of ways to perform a TR operation is less than $2n^2$ since we can choose any of the $n$ taxa and reinsert it at any of the $2n-5$ edges other than the one at which it was before the operation. This gives fewer than $(2n^2)^{c\log^4 n}$ phylogenies that can be created by repeating the operation $c \log^4 n$ times. Taking $b=4c$ finishes the proof.
\end{proof}

Now we can prove the maximum quartet compatibility consistency theorem.

\begin{proof}
Suppose some tree $T'$ is consistent with more quartets than $T^*$, and $d_Q(T',T^*) = q$.  At least half of the $q$ quartets where $T^*$ and $T'$ differ must be erroneous; since they are independent errors, this has probability at most $\exp(-q(\frac{(1-2p)^2}{2}))$ by the Chernoff bound.

Let $\mathcal{T}_0$ be the set of all incorrect phylogenies with quartet distance from $T^*$ less than $n \log^2 n$ .  Then $|{\mathcal T}^0| \le n^{b \log^4 n}$, and for trees in ${\mathcal T}_0$, Lemma \ref{dist-lemma} gives that $q\ge n-3$.  The probability that any tree in $T_0$ is consistent with more queries than $T^*$ is bounded by $n^{b \log^4 n} \exp(-(n-3)(\frac{(1-2p)^2}{2}))$, which is $o(1)$ as $n$ grows.

Now, consider the incorrect phylogenies ${\mathcal T}_1$ that are not in ${\mathcal T}_0$.  There are fewer than $2^n n! < 2^{n (1+\log n)}$ such topologies, and for each, $d_q(T,T^*) \ge n \log^2 n$.  The probability that any tree in ${\mathcal T}^1$ is consistent with more quartets than $T^*$ is bounded above by $2^{n(1+\log n)} \exp(-n\log^2 n(\frac{(1-2p)^2}{2}))$, which is $o(1)$ as $n$ grows.

So the probability that any incorrect tree is consistent with more quartets than $T^*$ converges to 0 as $n$ grows.
\end{proof}

\section{Experiments}

We have developed a prototype implementation of our algorithm to investigate its running time and properties. We have tested the algorithm in three scenarios. First, we tested the performance of the algorithm for the case with no errors. Second, we tested the performance of the random walk algorithm when the data was generated according to the model with independent errors. Finally, we ran the random walk algorithm on real biological datasets. 

The tree topologies used in the synthetic data sets were chosen at random from the uniform distribution. In the iid error case, every quartet query gave one of the two possible wrong answers with probability $p$. In our experiments, we set $p=0.1$.

The algorithm for error-free data is very fast even for reasonably large phylogenies. For data sets  having $10000$ taxa or less, constructing the tree takes less than a second. For $20000$ taxa, it takes roughly $2$ seconds. 

The random walk algorithm is roughly 5 times slower than the algorithm for error-free data. Constructing a tree having $10000$ taxa takes about $5$ seconds, whereas a tree with $20000$ taxa requires $9$ seconds.
\begin{center}
\begin{table}
\caption{The running times of the algorithm for the error-free and iid data sets}
\centering 
 \begin{tabular}{|c|c|c|c|c|}

		Algorithm & 1000 & 5000 & 10000 & 20000  \\ \hline
        Error-free & $<1s$ & $<1s$ & $<1s$ & $2s$ \\ \hline
        Random walk & $<1s$ & $2s$ & $5s$ & $9s$  \\ \hline
		\end{tabular}
      \end{table}
\end{center}

We ran the algorithm on several protein families from the Pfam database~\cite{Pfam}. Quartet queries were answered with the Four-Point method~\cite{Felsenstein} based on estimated evolutionary distances between sequences. Distances were estimated based on pairwise BLOSUM62 scores using a method by Sonnhammer and Hollich~\cite{Sonnhammer}. We used neighbor-joining trees on a subset of 150 sequences (chosen at random from the whole set of sequences) as our initial guide trees. Our prototype implementation was able to process a dataset of around 12000 sequences in about 16 minutes (see Table \ref{results}).

\begin{center}
\begin{table}
\caption{The running times of the algorithm for several Pfam families}
\centering
 \begin{tabular}{|c|c|c|c|}
	  \label{results}
      
		Protein family & Sequences & Average length & running time  \\ \hline
        Maf(PF02545) & 1980 & 189.60 & 38s \\ \hline
        2Oxoacid\_dh(PF00198) & 3701 & 225.10 & 1m49s  \\ \hline
		PALP(PF00291) & 11815 & 294.40 & 15m42s  \\ \hline
      \end{tabular}
      \end{table}
\end{center}

In all our experiments, the height of search trees constructed by the algorithm was less than $40$. This supports our view that the constants in Lemma \ref{tree-height-blah} can be improved. 
	  
\section{Conclusion}

We have presented a fast algorithm that is guaranteed to reconstruct the correct phylogeny with high probability under an error model where each quartet query errs with a fixed probability, independently of others. The algorithm runs in $O(n \log n)$ time, which is the lower bound for any phylogeny reconstruction algorithm. Our prototype implementation seems reasonably fast on both real and simulated datasets.

This work could be extended in many directions. From a theoretical perspective, it is interesting whether there exist fast algorithms that offer similar performance guarantees under commonly studied models of sequence evolution, such as Jukes-Cantor or Cavender-Farris. 

 From a practical perspective, it would be interesting to compare the results of our algorithm to others. We plan to extend our algorithm to make use of additional information such as the length of the middle edge in reconstructed quartets. This would enable the algorithm to distinguish between more credible and less credible queries, which may lead to an overall performance improvement. Another way to improve the algorithm is by improving the procedure of finding new quartets to ask so as to minimize the correlation between errors.
\newpage
\bibliographystyle{splncs03}
\bibliography{full_quartets4}
\end{document}